\documentclass[a4paper,12pt]{article}
\usepackage{geometry}
\usepackage[toc,page]{appendix}
\usepackage[colorlinks,citecolor=blue,linkcolor=blue, urlcolor=blue]{hyperref}
\usepackage{amsmath,amsthm,amssymb}
\usepackage{graphicx}
\usepackage{caption}
\usepackage{tikz,pgfplots}
\usepackage{subcaption}
\usepackage{cleveref}
\usepackage{sectsty}
\usepackage{mathtools}
\usepackage{color}
\usepackage{dsfont}
\usepackage{comment}
\usepackage[framemethod=tikz]{mdframed}
\usepackage{geometry}
\usepackage{enumitem}
\usepackage{algorithm}
\usepackage{algorithmic}
\theoremstyle{definition}
\newtheorem{theorem}{Theorem}[section]

\newtheorem{definition}[theorem]{Definition}
\newtheorem{lemma}[theorem]{Lemma}

\newtheorem{proposition}[theorem]{Proposition}

\newcommand*{\argmax}{\operatornamewithlimits{argmax}\limits}
\newcommand*{\argmin}{\operatornamewithlimits{argmin}\limits}
\newcommand{\BE}[1]{\mathbb{E}\left({#1}\right)}

\newcommand{\R}{\mathcal{R}}

\newcommand{\norm}[1]{\left\lVert#1\right\rVert}

\newcommand{\mud}{\bar{\mu}}
\newcommand{\sid}{\bar{\sigma}}

\newcommand{\setR}{\mathbb{R}}

\sectionfont{\fontsize{12}{15}\selectfont}
\subsectionfont{\fontsize{12}{15}\selectfont}
\numberwithin{equation}{section}
\date{\today}
\makeatletter 
\@addtoreset{equation}{section}

\makeatother 
\begin{document}
\title{ Risk management under Omega measure}
\maketitle

\begin{center}
\author{Michael R. Metel \footnote{Laboratoire de Recherche en Informatique, Universit\'e Paris-Sud, Orsay, France, metel@lri.fr; Corresponding author}  \and Traian A.~Pirvu\footnote{Department of Mathematics and Statistics, McMaster University, 1280 Main Street West, Hamilton, ON, L8S 4K1, Canada, tpirvu@math.mcmaster.ca} \and Julian Wong\footnote{Department of Mathematics and Statistics, McMaster University, 1280 Main Street West, Hamilton, ON, L8S 4K1, Canada, julianwwong@gmail.com}}
\end{center}
\date{}

\begin{abstract}
We prove that the Omega measure, which considers all moments when assessing portfolio performance, is equivalent to the widely used Sharpe ratio under jointly elliptic distributions of returns. Portfolio optimization of the Sharpe ratio is then explored, with an active-set algorithm presented for markets prohibiting short sales. When asymmetric returns are considered we show that the Omega measure and Sharpe ratio lead to different optimal portfolios.
\end{abstract}

\vspace{0.5cm}

 {\bf{Acknowledgements}} This work is supported by NSERC grant 371653-09 and by the Digiteo Chair C\&O program. 

\section{Introduction}
In the modern world of finance and insurance, it is routine for investors, firms and companies to manage different financial/insurance assets in the hope of increasing their capital gain. The collection of such investments is known as a portfolio, and it is designed to match the investor's preference. Different compositions of varying assets allow for a diversity of combinations that suit distinct appetites. For example, a bulge bracket investment bank such as J.P. Morgan is willing to undertake more risk to compensate for a larger return, in comparison to a retiree who is overseeing his retirement fund. However, despite an individual's taste, investors face the challenge of balancing reward and risk, as a high reward investment is often tightly linked with high underlying risk, and thus the main goal of portfolio management is finding the optimal tradeoff between the two.\\

The mean-variance portfolio model, proposed by Harry Markowitz~\cite{Marko} serves as the keystone to portfolio theory. He formalized the problem of a rational, risk adverse investor that faces the tradeoff between reward and risk as proposed above. In such a scenario, reward and risk are defined as the expected return from the portfolio and its variance. There are problems with the implementation of the Markowitz model when the universe of assets is large.
In this situation the assets' sample covariance matrix is not an efficient estimator of the assets' true covariance matrix. Therefore, using the sample mean and covariance matrix in the mean-variance optimization procedure will result in an optimal return estimate different from reality. A fix for this problem
is proposed in \cite{Bai}, by using the theory of the large-dimensional random matrix. Another reason for the poor performance of the optimal mean-variance portfolio is perhaps due to the symmetry of asset returns. \cite{Low1} shows that it is possible to enhance mean-variance portfolio selection by allowing for distributional asymmetries. Portfolio optimization under skewed returns is performed in several papers such as in \cite{Hu} and \cite{Low2}.\\

Under the mean-variance framework, various major portfolio theories have sprouted, and one of the major developments proposed by William F. Sharpe~\cite{Sharpe} is known as the Sharpe ratio. The Sharpe ratio is the most fundamental of performance measures, which are critical in the evaluation, management and trading of portfolios. Under the mean-variance portfolio framework, the Sharpe ratio compares the return of the portfolio with the risk-free interest rate, which serves as a significant benchmark, owing to the fact that if overall return of the portfolio ranks below the risk free rate, investors should put their capital in the money market and earn interest without bearing any risk.
The Sharpe ratio is greatly incorporated as a modern investment strategy, and is highly appraised by investors. However, the Sharpe ratio only comprises and examines the first two moments of the return distribution, namely the expected return and the variance in return, while distribution properties such as skewness and kurtosis, which measure asymmetry and thickness of the tail distribution at the third and fourth moments respectively, may profoundly impact the performance of the portfolio. \cite{{DeMiguel}} compares the optimal mean-variance portfolio with the naive $\frac{1}{N}$ portfolio. They found that the $\frac{1}{N}$ rule performs better than the optimal mean-variance portfolio in terms of the  Sharpe ratio, indicating that the gain from optimal diversification is higher when compared to the offset produced by estimation error.\\

The failure of the Sharpe ratio to address higher moments motivated Shadwick and Keating~\cite{W} to develop the Omega measure, which captures all moments of the return distribution, including the expected value and variance. The Omega measure serves as a universal performance measure as it can be applied to any portfolio that follows a well-defined return distribution.\\

Even though the Omega measure was developed over 10 years ago, little research has been done to address its compatibility with previous developments, namely with distribution functions that only involve lower moments. This paper aims to explore and address the backward compatibility of the Omega measure. We consider a market (financial or insurance) encompassing several risks within a one period paradigm. The risks are first assumed to follow a jointly elliptical distribution. Under this framework we prove that the Sharpe ratio and the Omega measure yield the same optimal portfolios. Next, Sharpe ratio portfolio optimization is explored. The quasi-concavity of the Sharpe ratio is employed to develop an active-set algorithm for markets banning short sales. The convergence of this algorithm is established and numerical results are presented. Moreover, we show that in a model with asymmetric returns the optimal Sharpe ratio portfolio fails to be optimal when Omega measure is considered.\\

The remainder of this paper  is organized as follows: In Section 2 we present the model. Section 3 provides the Sharpe ratio and Omega measure equivalence
within the class of elliptical distributions of returns. Portfolio optimization formulations are presented in Section 4. Numerical analysis is performed in Section 5, with numerical results displayed in Section 6. Section 7 presents a model with asymmetric returns. The conclusion is summarized in Section 8. The paper ends with an Appendix containing the proofs.

\section{The Model}

We have a market (financial or insurance) model which encompasses several instruments denoted $S_{1},...,S_{n}.$
We consider a single period model from time $t=0$ to $t=1$. For each instrument, let the arithmetic return be
\begin{equation}
R_{i}=\frac{S_{i}(1)-S_{i}(0)}{S_{i}(0)},\nonumber
\end{equation}
and
$$ \R=(R_{1}, R_{2},\cdots, R_{n}). $$
We assume the return of the portfolio follows an \emph{elliptically symmetric} distribution. Then the vector of means $E(\R)=\mu=(\mu_1,...,\mu_n)^T$  and the $n \times n$ covariance matrix $Cov(\R)=\Sigma=(\sigma_{ij})_{i,j}$ exist, and we further assume that $\Sigma$ is invertible. The density $f,$ if it exists, is
$$f(x)=|\Sigma|^{-\frac{1}{2}}g[(x-\mu)^T\Sigma^{-1}(x-\mu)],$$
where $x\in \setR^n$ and $g:\setR^+\rightarrow\setR^+$ is called the \emph{density generator} or \emph{shape} of $R$, and we write
 $$ \R\sim EC_{n} (\mu, \Sigma;g),  $$
where $(\mu,\Sigma)$ is called the parametric part and $g$ is called the non-parametric part of the elliptical distribution. The characteristic function $\psi$ of $R$ is
\begin{equation}\label{cf}
 \psi_{\R}  (t)=E \exp{(i t^T \R)}=  \exp{(i t^{T} \mu)} \phi(t^{T} \Sigma t),
 \end{equation}
for some scalar function $\phi$, called the \emph{characteristic generator}. For background on the elliptically symmetric distribution, which is also called elliptically countered, see \cite{fang}, and \cite{{bingham2002semi}}.\\

The class of elliptical distributions, which have densities and defined mean and covariance is rich enough to contain several common distributions of asset returns: the multivariate normal distribution, the multivariate $t$ distribution, normal-variance mixture distributions, symmetric stable distributions, the symmetric generalized hyperbolic distribution, the symmetric variance-gamma distribution, and the multivariate exponential power family (and thus the Laplace distribution).
One advantage of this class is that the non-parametric part $g$ ''escapes the curse of dimensionality'' cf \cite{{bingham2002semi}}.
This class is chosen to model the stock returns by \cite{Chamberlain}, \cite{Owen}, and \cite{{bingham2003semi}}.\\

Elliptical distributions are appealing for portfolio analysis, since it is a closed class under linear combinations. A portfolio at times $t=0$  and $t=1$ will respectively be
\begin{align*}
X(0) &= \Delta _{1}S_{1}(0)+\cdots +\Delta _{n}S_{n}(0)\\
X(1) &= \Delta _{1}S_{1}(1)+\cdots +\Delta _{n}S_{n}(1)
\end{align*}

Let the arithmetic return of the portfolio be
$$
R=\frac{X(1)-X(0)}{X(0)}.$$
The following Lemma gives the distribution of $R.$
\begin{lemma}\label{L1}
Let $$w_i = \frac{\Delta _iS_i(0)}{\Delta _1S_1(0)+\cdots +\Delta _nS_n(0)}$$ be the proportion of the initial wealth invested in instrument i,
and $w$ be the vector with components $w_i$. Then $R$ follows an elliptical distribution
\begin{equation}
R \sim  EC_{1} (\mud, \sid; g)\nonumber
\end{equation}
where
\begin{equation}\label{eue}
\mud:= w \cdot \mu =\sum_{i=1}^n  w_i \mu_i,\quad \sid^2:=w^T\Sigma \, w=\sum_{i=1}^n \sum_{j=1}^n w_i w_j \sigma_{ij}.
\end{equation}

\end{lemma}
\begin{proof}
See The Appendix
\end{proof}

Let us consider the Sharpe Ratio and Omega measure defined by the formal definitions.
\begin{definition}
The Sharpe ratio of a portfolio with return R is defined as \begin{equation}
S(R) = \frac{\mud-r_f}{\sid}\nonumber\end{equation}where $\mud$ is the expected return of the portfolio, $\sid$ is the standard deviation of return, and $r_f$ is the risk-free interest rate.
\end{definition}
\begin{definition}
The Omega measure of a portfolio with return R is defined as
\begin{equation} \Omega (R ) = \frac{\int_{L}^{\infty}(1-F(x))dx}{\int_{-\infty }^{L}F(x)dx}, \nonumber\end{equation}
where $F(x)$ is the cumulative distribution function of the return distribution R, and $L$ is an exogenously satisfied benchmark index. \end{definition}

The intuition behind the Omega measure is simple; by selecting a benchmark $L$, which serves as a reference that our portfolio is aiming to beat, the Omega measure compares the area of the cumulative distribution function from the right of L to the area to the left of L. Under such a definition, the Omega measure encompasses the entire return distribution, therefore incorporating higher moment properties as discussed.

\section{Sharpe Ratio and Omega Measure Equivalence}

When holding a portfolio, an investor uses a performance measure such as the Sharpe ratio or the Omega measure to evaluate how well the portfolio is performing. Hence it is a natural question to ask how one should distribute his wealth in order to maximize his portfolio under the Omega measure. The following theorem states that using the Sharpe ratio or the Omega measure to optimize portfolio performance leads to the same optimal portfolio within the class of elliptical distributions of returns.

\begin{theorem}\label{main}
Recall that under our framework the portfolio return $R$ is elliptically distributed $R \sim EC_{1} (\mud, \sid; g)$. If
$r_f = L$ we claim that $$\underset{w_1,..,w_n}{\max}\Omega(R)$$ is equivalent to $$\underset{w_1,..,w_n}{\max}S(R).$$
\end{theorem}
\begin{proof}
See The Appendix
\end{proof}

\section{Portfolio Optimization}

Given Theorem \ref{main}, we are able to transform optimization problems of the Omega measure into optimization problems of the Sharpe ratio for elliptical distributions. Let $e=\mu-L$, the excess expected return above a selected benchmark index $L$. With no restrictions on short selling, our optimization problem is as follows.

\begin{alignat}{6}
\label{4.1}
\tag{4.1}
&\max&&\text{ }\frac{w^Te}{\sqrt{w^T\Sigma w}}\nonumber\\
&\mbox{s.t. }&&\sum_{i=1}^{n}w_i=1\nonumber
\end{alignat}

However, certain financial markets prohibit the act of short selling, especially during periods of financial upheaval. An example would be the U.S. securities market under the 2008 financial crisis, when the U.S. Securities and Exchange Commission prohibited the act of short selling to protect the integrity of the securities market. Hence we are also interested in the following problem as well.

\begin{alignat}{6}
\label{4.2}
\tag{4.2}
&\max&&\text{ }\frac{w^Te}{\sqrt{w^T\Sigma w}}\nonumber\\
&\mbox{s.t. }&&\sum_{i=1}^{n}w_i=1\nonumber\\
&&&w_i\geq 0\hspace{25 pt}i=1,...,n\nonumber
\end{alignat}

\section{Numerical Analysis}

The optimal solution to (\ref{4.1}) can be found directly as described in the following proposition.

\begin{proposition}\label{ecs}
The optimal solution to (\ref{4.1}) is
$w^*=\frac{\hat{w}}{\sum_i^n \hat{w}_i}$, where $\hat{w}=\Sigma^{-1}e$.
\end{proposition}
\begin{proof}
See The Appendix
\end{proof}

We require the following properties of the Sharpe ratio in developing an algorithm for solving (\ref{4.2}).

\begin{proposition}\label{quasiconcave}
The Sharpe ratio $S(w)=\frac{w^Te}{\sqrt{w^T\Sigma w}}$ is a quasi-concave function and $\nabla S(w)= 0$ iff $w=c \Sigma^{-1}e$ for some $c\neq 0$.
\end{proposition}
\begin{proof}
See The Appendix
\end{proof}

If $\Sigma^{-1}e\geq 0$ then our optimal solution for (\ref{4.1}) is also optimal for (\ref{4.2}), so let us assume that for (\ref{4.2}), our optimal solution $w^*\neq c \Sigma^{-1}e$ for any $c$. By our assumption, $\nabla S(w^*)\neq 0$ and the following theorem is applicable.

\begin{theorem}[Arrow \& Enthroven \cite{Arrow1961}]
Let $f(x)$ be a differentiable quasi-concave function subject to non-negativity constraints. If $\nabla f(x^*)\neq 0$ and $x^*$ satisfies the KKT conditions with constants $\mu^*$, then it is a global optimal solution.
\end{theorem}

The KKT conditions for (\ref{4.2}), ignoring the equality constraint are as follows, where $\nabla S(w)=\frac{e}{\sqrt{w^T\Sigma w}}-\frac{w^Te\Sigma w}{(w^T\Sigma w)^{\frac{3}{2}}}$.

\begin{align*}
&\frac{e}{\sqrt{w^T\Sigma w}}-\frac{w^Te\Sigma w}{(w^T\Sigma w)^{\frac{3}{2}}}+\mu=0&&\text{ (stationarity)}\label{KKT}\tag{5.1}\\
&w\geq 0&&\text{ (primal feasibility)}\\
&\mu\geq 0&&\text{ (dual feasibility)}\\
&\mu^Tw=0&&\text{ (complementary slackness)}
\end{align*}

Consider the sets $P$ and $W$ defined by

$$P:= \{i\in\{1,\ldots, n\} :w_i>0\}$$

$$W:=\{i\in\{1,\ldots, n\} :w_i=0\}$$

Let us permute the data so that $$e=[e_P;e_W],\qquad w=[w_P;w_W],$$ and let $\Sigma_P$ be the covariance matrix of the instruments indexed by $P$. Let $|P|$ be the number of elements of $P.$ At optimality, the first $|P|$ rows of (\ref{KKT}) will equal $$e_P-\frac{w^T_Pe_P\Sigma_Pw_P}{w^T_P\Sigma_Pw_P}=0,$$ with solution $$w_P=c\Sigma^{-1}_Pe_P,\,\, \mbox{for}\,\, c\neq 0.$$ Therefore, the optimal solution of (\ref{4.2}) is the optimal solution of (\ref{4.1}) for some unknown subset of instruments $P$.
For ease in what follows, we will always take $c=1$.\\

Our main objective then is to find the optimal set $P$, after which the optimal solution can be found by solving a positive definite system of linear equations. We propose the use of the following active-set algorithm to solve (\ref{4.2}), which is inspired by Algorithm 16.3 in~\cite{Noc}.

\begin{algorithm}[H]                     
\caption{Sharpe Ratio active-set (SRAS) algorithm}             
\label{alg1}                           
\begin{algorithmic}[1]                 
    \STATE $i=0$
    \STATE $w^i=\mathbf{0}$
    \STATE $j=\argmax\frac{e_j}{\sqrt{\Sigma_{jj}}}$
    \STATE $w^i_j=\frac{e_j}{\sqrt{\Sigma_{jj}}}$
    \STATE $W^i=\{j | w^i_j=0\}$
    \STATE $P^i=\{j | w^i_j>0\}$
    \LOOP{}
    \STATE $x^i_{P^i}=\Sigma_{P^i}^{-1}e_{P^i}$
    \STATE $x^i_{W^i}=\mathbf{0}$
    \STATE $p^i=x^i-w^i$
    \IF{$p^i=0$}
    \STATE $\mu^i_{W^i}=\frac{w^{iT}e(\Sigma w^i)_{W^i}}{(w^{iT}\Sigma w^i)^{\frac{3}{2}}}-\frac{e_{W^i}}{\sqrt{w^{iT}\Sigma w^i}}$
    \IF{$\mu^i_j \geq 0 \text{ }\forall j\in W^i$}
    \STATE $w^*=\frac{w^i}{\sum_{j=1}^{n} w^i_j}$
    \STATE \textbf{quit}
    \ELSE
    \STATE $k^i=\argmin\limits_{j\in W^i} \mu^i_j$
    \STATE $W^{i+1}=W^i\setminus\{k^i\}$
    \STATE $P^{i+1}=P^i\cup\{k^i\}$
    \STATE $w^{i+1}=w^i$
    \ENDIF
    \ELSE
    \STATE $\alpha^i=\min\{1,\min\limits_{j \in P^i, p^i_j<0} \frac{-w^i_j}{p^i_j}\}$
    \IF{$\alpha^i<1$}
    \STATE $h^i=\argmin\limits_{j \in P^i, p^i_j<0} \frac{-w^i_j}{p^i_j}$
    \STATE $W^{i+1}=W^i\cup\{h^i\}$
    \STATE $P^{i+1}=P^i\setminus\{h^i\}$
    \ENDIF
    \STATE $w^{i+1}=w^i+\alpha^i p^i$
    \ENDIF
    \STATE $i=i+1$
    \ENDLOOP
\end{algorithmic}
\end{algorithm}

We find the portfolio consisting of a single instrument which maximizes the Sharpe ratio to initialize the algorithm in lines 1-6.
At iteration $i$, $x^i_{P^i}$ is set to maximize the Sharpe ratio, which in general is not feasible in (\ref{4.2}), in line 8. If the current feasible solution $w^i=x^i$, we check if $\mu^i_{W^i}\geq 0$. If so, then $$w^*=\frac{w^i}{\sum_{j=1}^{n} w^i_j}$$ is the optimal solution to (\ref{4.2}), or else we remove the index of the minimum value of dual variables $\mu^i$ from $W^i$ to form $W^{i+1}$ in lines 11-21. If $w^i\neq x^i$, $w^{i+1}$ is set by moving in the direction of $x^{i}$ from $w^{i}$ while remaining feasible in (\ref{4.2}). If $w^{i+1}\neq x^i$, the index of the first blocking constraint $j\in P^i$ is added to $W^i$ to create $W^{i+1}$ in lines 22-30.

\begin{theorem}\label{con}
The SRAS algorithm is convergent.
\end{theorem}
\begin{proof}
See The Appendix
\end{proof}

There is in fact a quadratic convex reformulation of this problem, see \cite{Bien}, which has the following formulation under the mild condition that there exists at least one stock with $e_i>0$, where $z>0$ is a free constant.

\begin{alignat}{6}
&\min&&\text{ }w^T\Sigma w\label{QF}\tag{5.2}\\
&\mbox{s.t. }&&w^Te=z\nonumber\\
&&&w_i\geq 0\nonumber
\end{alignat}

After solving, the $w_i$ simply have to be normalized to sum to one to obtain the optimal solution. The choice of $z$ can affect solution quality, in particular when the number of instruments $n$ becomes large and the algorithm used to solve (\ref{QF}) employs a stopping criteria of the form $|w^i-w^{i+1}|\leq \text{tolerance}$. In practice we have found choosing $z=e^T1$, ensuring the average value of elements of $w^i$ equals 1, gives high quality solutions with virtually no optimality gap compared to the active-set algorithm, without having to alter default stopping criteria.

\section{Numerical Results}

A computational experiment was conducted where the SRAS algorithm was compared to Gurobi 7.0 using data derived from historical stock prices from two stock market indices. All computing was done using Matlab R2016a on a Windows 10 64-bit, AMD A8-7410 processor with 8 GB of RAM.\\

Six problems were used for testing. Historical stock prices of the Dow Jones Industrial Average and the S\&P/TSX 60 were used to calculate the expectation and covariance of instrument returns. For each index, the past year, 2 years and 5 years were used for estimation. This data was generated using the website InvestSpy \cite{IS}. Results are presented in Table \ref{T10} below. We observe that the mean computing time of SRAS is over an order of magnitude faster when compared to Gurobi. Another positive aspect of the SRAS algorithm is its relative simplicity compared to the quadratic programming reformulation, which is generally solved by using an interior point method.\\

\begin{table}[htb]
\centerline{
\resizebox{0.65\textwidth}{!}{
\begin{tabular}{lcccccc}
 &\multicolumn{2}{c}{SRAS} &\multicolumn{2}{c}{Gurobi}\\
\hline
 &Time (s)& Solution&Time (s)&Solution\\
  \hline
Dow 1 Yr &0.0386 &2.6769 &0.6881 &2.6769\\
Dow 2 Yr &0.0057 &3.3883 &0.5551 &3.3883\\
Dow 5 Yr & 0.0030 &15.7604&0.5829&15.7604\\
S\&P 1 Yr& 0.0479 &7.2073&0.6569 &7.2073\\
S\&P 2 Yr& 0.0095 &4.4550&0.6233 &4.4550\\
S\&P 5 Yr& 0.0053 &5.0557&0.5562 &5.0557\\
\hline
Mean     & 0.0184 & &0.6104 &
\end{tabular}}}
\caption{Numerical results} \label{T10}
\end{table}

\section{Model with Skewness}

We show numerically that the Omega measure is not equivalent to the Sharpe ratio for skewed distributions. Our estimation of Omega measure uses the following proposition.

\begin{proposition}\label{skew}
The Omega measure is equal to $\frac{\bar{\mu}-L}{\BE{\max(L-R,0)}}-1,$ i.e.,
$$  \Omega (R ) = \frac{\bar{\mu}-L}{\BE{\max(L-R,0)}}-1. $$
\end{proposition}
\begin{proof}
See The Appendix
\end{proof}

We consider the skew-normal distribution \cite{azzalini2005skew} which is closed under affine transformation and has probability distribution function

$$f(r)=\frac{2}{\omega}\phi(\frac{r-\epsilon}{\omega})\Phi(\alpha (\frac{r-\epsilon}{\omega})).$$

where $\phi(\cdot)$ and $\Phi(\cdot)$ are the standard normal probability distribution function and cumulative distribution function respectively, with location paramter $\epsilon$, scale $\omega$ and shape $\alpha$.\\

For a given skewness $\gamma_1$, let $$|\delta|=\sqrt{\frac{(\pi/2)|\gamma_1|^{2/3}}{((4-\pi)/2)^{2/3}+|\gamma_1|^{2/3}}},$$ where the sign of $\delta$ is chosen negative for left skewness and positive for right skewness. Given $\delta$,
$$\alpha=\frac{\delta}{\sqrt{1-\delta^2}},$$
and for a desired standard deviation,  $$\omega=\frac{\sigma}{\sqrt{1-2\delta^2/2}},$$ and mean, $$\epsilon=\mu-\omega\delta\sqrt{2/\pi}.$$\\

We plot the Omega measure for $L=0.01$, $\mu=0.1$ and $\sigma=0.3$, with $\gamma_1$ varying over the domain $[-0.99,0.99]$ in increments of 0.01. Monte Carlo integration was used to estimate $\BE{\max(L-R,0)}$ by taking 10 million samples of $R$ and taking the mean of $\max(L-R,0)$.

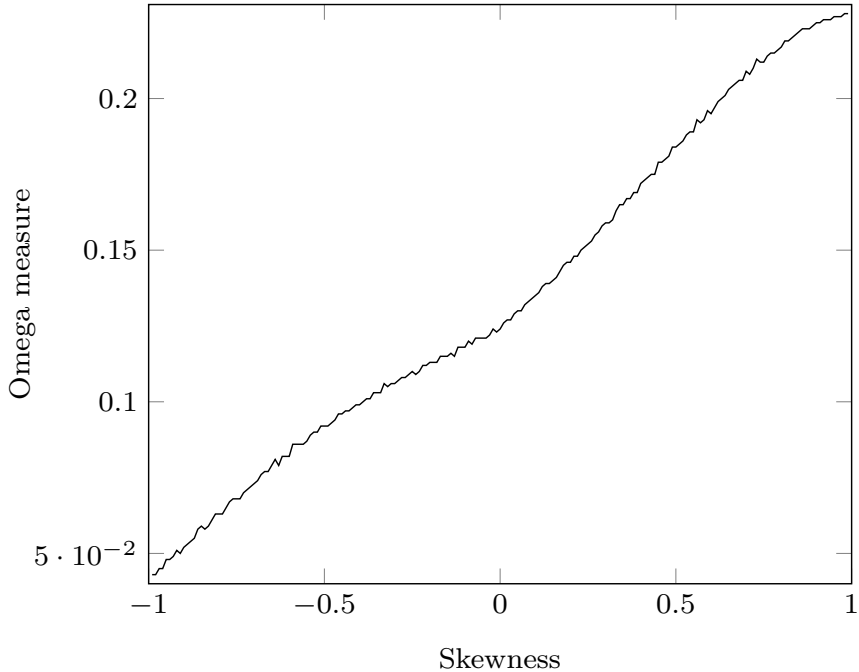
\begin{figure}[H]
 \centerline{
\resizebox{0.75\textwidth}{!}{
\begin{tikzpicture}
\scriptsize
\begin{axis}[xlabel=Skewness,ylabel=Omega measure, name=plot1,xmin=-1, xmax=1, ymin=0.04, ymax=0.231]
                \addplot[mark=none,mark size=1,draw=black]
                table[x=G,y=O]
                {result2.dat};
    \end{axis}
\end{tikzpicture}}}
\caption{Omega measure versus skewness for a skew-normal random variable with $\mu=0.1$, $\sigma=0.3$ and $L=0.01$.} \label{T1}
\end{figure}

Under the Sharpe ratio we are indifferent to all of the plotted portfolios, each having a Sharpe ratio $S(R)=0.3$, but under the Omega measure, taking into consideration higher moments, it is clear that we would prefer a portfolio with right skewness in this example.

\section{Conclusion and Future Work}
In this paper, we have proved the equivalence of the Omega measure and the Sharpe ratio under jointly elliptical distributions of returns. The portfolio optimization of the Sharpe ratio with and without short sales was numerically analyzed. An active-set algorithm was presented for markets prohibiting short sales, with an improvement in average solution time of over an order of magnitude when compared to standard optimization techniques. Numerical experiments show that when the return distributions are not symmetric the Omega measure and the Sharpe ratio are not equivalent. Future research could be done to develop optimization methods for the Omega measure under more general distribution assumptions such as the skew-elliptical distribution.

\bibliographystyle{plain}
\bibliography{OmegaOpt}

\section{Appendix}

\begin{proof} of Lemma \ref{L1}

\begin{align*}
R&=\frac{\Delta _{1}S_1(1)+\cdots +\Delta _{n}S_n(1)-(\Delta _{1}S_1(0)+\cdots +\Delta _{n}S_n(0))}{\Delta _{1}S_1(0)+\cdots +\Delta _{n}S_n(0)} \nonumber\\
&=\frac{\Delta _{1}S_1(1)-\Delta_{1}S_1(0)}{\Delta _{1}S_1(0)+\cdots +\Delta _{n}S_n(0)}+\cdots+\frac{\Delta _{n}S_n(1)-\Delta_{n}S_n(0) }{\Delta _{1}S_1(0)+\cdots +\Delta _{n}S_n(0)} \nonumber\\
&=S_1(1)-S_1(0)\left(\frac{\Delta _{1}}{\Delta _{1}S_1(0)+\cdots +\Delta _{n}S_n(0)}\right)+\cdots+S_n(1)\nonumber\\&-S_n(0)\left(\frac{\Delta _{n}}{\Delta _{1}S_1(0)+\cdots +\Delta _{n}S_n(0)}\right) \nonumber\\
&=\frac{S_1(1)-S_1(0)}{S_1(0)}\left(\frac{\Delta _1S_1(0)}{\Delta _1S_1(0)+\cdots +\Delta _nS_n(0)}\right)+\cdots \nonumber\\&+ \frac{S_n(1)-S_n(0)}{S_n(0)}\left(\frac{\Delta _nS_n(0)}{\Delta _1S_1(0)+\cdots+\Delta _nS_n(0)}\right) \nonumber\\
&= w_{1}R_{1}+\cdots+w_{n}R_{n}
\end{align*}
Thus the return of the portfolio R is a linear combination of $R_1,\cdots,R_n$. The closedness of elliptical distributions under linear combinations yields the claim.
\end{proof}

\begin{proof} of Theorem \ref{main}\\

Under our framework, we are now able to simplify the Omega measure. Recall the Omega measure is defined as
\begin{equation*}
\Omega (R ) = \frac{\int_{L}^{\infty}(1-F(x))dx}{\int_{-\infty }^{L}F(x)dx}.
\end{equation*}
Here $F(x)=\int_{-\infty}^{x}f(r)dr$ is the cumulative distribution function of the portfolio with arithmetic return R, with probability distribution function $f(r)$.
Thus
\begin{align}
\Omega(R)&=\frac{\int_{L}^{\infty }\left(1- \int_{-\infty}^{x}f(r)dr \right)dx}{\int_{-\infty}^{L}   \int_{-\infty}^{x}f(r)drdx}\nonumber \\
&=\frac{\int_{L}^{\infty }\int_{x}^{\infty }   f(r)dr  dx}{\int_{-\infty}^{L}\int_{-\infty}^{x }  f(r)dr dx}\nonumber
\end{align}
We use Fubini's theorem to change the order of integration. Let $D_{1}$ be the integration region of the integral in the numerator and let $D_{2}$ be the integration region of the integral in the denominator, then
\begin{equation}
D_{1} = \left \{(x,r)|x\in (L,\infty),r\in(x,\infty) \right \} =\left\{(x,r)|x\in(L,r),r\in(L,\infty)\right\} \nonumber
\end{equation}
and
\begin{equation}
D_{2} = \left \{(x,r)|x\in (-\infty,L),r\in(-\infty,x) \right \} =\left\{(x,r)|x\in(r,L),r\in(-\infty,L)\right\} \nonumber
\end{equation}

Thus, by Fubini's Theorem
\begin{equation}\label{eq:omega}
\Omega(R )=\frac{\int_{L}^{\infty }\int_{L}^{r } f(r) dxdr}{\int_{-\infty}^{L}\int_{r}^{L }        f(r) dxdr}
\end{equation}

Under elliptical distributions $$f(r)=\frac{1}{ \bar{\sigma }} g (\left(\frac{r-\bar{\mu }}{\bar{\sigma}} \right)^2 ).$$ Evaluating the upper integral gives us
\begin{align*}
\int_{L}^{\infty }\int_{L}^{r } \frac{1}{ \bar{\sigma }} g (\left(\frac{r-\bar{\mu }}{\bar{\sigma}} \right)^2 )    dxdr &= \frac{1}{ \bar{\sigma }}  \int_{L}^{\infty }(r-L) g (\left(\frac{r-\bar{\mu }}{\bar{\sigma}} \right)^2 )   dr \nonumber \\
&= \frac{1}{ \bar{\sigma }}  \int_{L}^{\infty }r g (\left(\frac{r-\bar{\mu }}{\bar{\sigma}} \right)^2 ) dr-\frac{L}{\bar{\sigma }}\int_{L}^{\infty } g (\left(\frac{r-\bar{\mu }}{\bar{\sigma}} \right)^2 ) dr. \nonumber\\
\intertext{Let us perform the change of variable. Therefore, we let $u=\frac{r-\bar{\mu }}{\bar{\sigma }}$, then $\bar{\sigma }du = dr$ and $r = \bar{\sigma }u+\bar{\mu}$}
\int_{L}^{\infty }\int_{L}^{r } \frac{1}{ \bar{\sigma }} g (\left(\frac{r-\bar{\mu }}{\bar{\sigma}} \right)^2 )  dxdr &=\int_{\frac{L-\bar{\mu}}{\bar{\sigma}}}^{\infty }(\bar{\sigma }u+\bar{\mu}) g(u^2)du-{L}\int_{\frac{L-\bar{\mu}}{\bar{\sigma}}}^{\infty } g(u^2)du \\
&=\bar{\sigma}\int_{\frac{L-\bar{\mu}}{\bar{\sigma}}}^{\infty }u g(u^2)du+(\bar{\mu}-L)\int_{\frac{L-\bar{\mu}}{\bar{\sigma}}}^{\infty } g(u^2)du
\end{align*}

Thus
$$
\int_{L}^{\infty }\int_{L}^{r }  \frac{1}{ \bar{\sigma }} g (\left(\frac{r-\bar{\mu }}{\bar{\sigma}} \right)^2 )   dxdr $$$$=\bar{\sigma}  \left[-\frac{1}{2} H_1  (\left(\frac{L-\bar{\mu }}{\bar{\sigma}} \right)^2 ) + \left(\frac{L-\bar{\mu }}{\bar{\sigma}} \right) H_2 \left(\frac{L-\bar{\mu }}{\bar{\sigma}} \right) -K  \left(\frac{L-\bar{\mu }}{\bar{\sigma}} \right)   \right],
$$
where
$$ H'_1(x)=g(x),\,\,\, H'_2(x)=g(x^2),\,\,\, K=\int_{-\infty}^ {\infty} g(x^2)\,dx.  $$
We use the same methodology for the lower integral to obtain

$$
 \int_{-\infty}^{L}\int_{r}^{L } \frac{1}{ \bar{\sigma }} g (\left(\frac{r-\bar{\mu }}{\bar{\sigma}} \right)^2 ) dxdr=\bar{\sigma}  \left[-\frac{1}{2} H_1  (\left(\frac{L-\bar{\mu }}{\bar{\sigma}} \right)^2 ) + \left(\frac{L-\bar{\mu }}{\bar{\sigma}} \right) H_2 \left(\frac{L-\bar{\mu }}{\bar{\sigma}} \right)     \right].
$$

Let $z=\frac{L-\bar{\mu}}{\bar{\sigma}}$, then
$$
\Omega(R) = G(z),
$$
where
$$
G(z)=1-\frac{K z}{  z H_2(z) -\frac{1}{2} H_1(z^2) }
$$
We claim that $\Omega(R)$ is a decreasing function of $z$. To see that, we first take the derivative of $G(z)$
$$
G'(z) =-\frac{1 }{2}\frac{K H_1(z^2)}{ \left( z H_2(z) -\frac{1}{2} H_1(z^2)  \right)^2 } \leq 0,
$$
since
$$H_1(x)=\int_{-\infty}^{x} g(u)\,du \geq 0, $$
due to the positivity of $g.$ Therefore $\Omega(R)$ is a decreasing function of $z.$ Hence
\begin{align*}
\underset{w_1,...,w_n}{\max} \Omega(R)
&\Leftrightarrow \underset{w_1,..,w_n}{\min}\frac{L-\bar{\mu}}{\bar{\sigma}}\\
&\Leftrightarrow \underset{w_1,..,w_n}{\max}\frac{\bar{\mu}-L}{\bar{\sigma}}\\
&\Leftrightarrow \underset{w_1,..,w_n}{\max}S(R)
\end{align*}
Therefore, maximizing the Omega measure over $\{w_1,...,w_n\}$ is equivalent to maximizing the Sharpe ratio over $\{w_1,...,w_n\}$ with risk-free interest rate $R_f$ equal to $L$.

\end{proof}

\begin{proof} of Proposition \ref{ecs}\\

The extended Cauchy-Schwarz inequality, see~\cite{johnson}, states that for vectors $b$ and $d$, and positive definite matrix $B$, $(b^Td)^2\leq (b^TBb)(d^TB^{-1}d)$ with equality if and only if $b=cB^{-1}d$ for any constant $c$. It follows that for the objective of (\ref{4.1}), $\frac{w^Te}{\sqrt{w^T\Sigma w}}\leq \sqrt{e^T\Sigma^{-1}e}$ for $w\neq 0$, with the maximum attained by $\hat{w}=\Sigma^{-1}e$. In order to satisfy the constraint $\sum_{i=1}^{n}w_i=1$, $\hat{w}$ is multiplied by the normalizing constant, $c=\frac{1}{\sum_i^n \hat{w}_i}$ to obtain the optimal solution to (\ref{4.1}).

\end{proof}

\begin{proof} of Proposition \ref{quasiconcave}\\

A function $f(x)$ is quasi-concave if its upper level sets $\{x | f(x)\geq t\}$ are convex. The upper level sets of $S(w)$ form second order conic constraints, $w^Te\geq t\sqrt{w^T\Sigma w}$, which define convex regions.\\

If $\nabla S(w)=\frac{e}{\sqrt{w^T\Sigma w}}-\frac{w^Te\Sigma w}{(w^T\Sigma w)^{\frac{3}{2}}} = 0$, then $w=c \Sigma^{-1}e$ for $c=\frac{w^T\Sigma w}{w^Te}$. Taking $w=c \Sigma^{-1}e$ for any  $c\neq 0$, it follows directly that $\nabla S(w)=0$.

\end{proof}

\begin{proof} of Theorem \ref{con}\\

\begin{lemma}\label{xk}
If $w^i=x^i$ but $w^i$ is suboptimal for (\ref{4.2}), then $\alpha^{i+1}>0$ in the next iteration, unless there exists a $j\in P^i$ such that $w^i_j=0$ and $p^{i+1}_j<0$.\\

\end{lemma}

\begin{proof} of Lemma \ref{xk}\\

Let $P^{i+1}=P^i\cup\{k\}$ for some $k\in W^i$. In the $i^{th}$ iteration, consider the rows of $P^{i+1}$ in (\ref{KKT}), $e_{P^{i+1}}-\frac{w^{iT}e\Sigma_{P^{i+1}}w^i_{P^{i+1}}}{w^{iT}\Sigma w^i}+\hat{\mu}^i_{P^{i+1}}=0$, where $\hat{\mu}^i=\sqrt{w^{iT}\Sigma w^i}\mu^i$. Given $w^i=x^i$, it follows that $w^{iT}e=w^{iT}\Sigma w^i$, and since $\hat{\mu}^i_{P^i}=\mathbf{0}$, we get $\Sigma_{P^{i+1}}w^i_{P^{i+1}}=\left[\begin{array}{c}e_{P^i}\\e_k+\hat{\mu}^i_k\\\end{array}\right]$. Taking the Cholesky decomposition, $\Sigma_{P^{i+1}}=LL^T$, we can write $L\left[\begin{array}{c}y^i_{P^i}\\y^i_k\\\end{array}\right]=\left[\begin{array}{c}e_{P^i}\\e_k+\hat{\mu}^i_k\\\end{array}\right]$ with $L^Tw^i_{P^{i+1}}=y^i$. Since $L^T$ is upper triangular, $L_{kk}w^i_k=y^i_k$, and so $y^i_k=0$. Considering now the $i+1^{th}$ iteration, $\Sigma_{P^{i+1}}x^{i+1}_{P^{i+1}}=\left[\begin{array}{c}e_{P^i}\\e_k\\\end{array}\right]$, and similarly if $L^Tx^{i+1}_{P^{i+1}}=y^{i+1}$ then $L\left[\begin{array}{c}y^{i+1}_{P^i}\\y^{i+1}_k\\\end{array}\right]=\left[\begin{array}{c}e_{P^i}\\e_k\\\end{array}\right]$. Since $L$ is lower triangular, $y^{i+1}_{P^i}=y^i_{P^i}$, $y^{i+1}_k=y^i_k-\frac{\hat{\mu}^i_k}{L_{kk}}$, and so $x^{i+1}_k=\frac{-\hat{\mu}^i_k}{L^2_{kk}}=-\hat{\mu}^i_k\Sigma^{-1}_{kk}$. As $w^i$ is not optimal, $\hat{\mu}^i_k<0$, so $p^{i+1}_k>0$. Assuming now $p^{i+1}_j\geq0$ for all $j\in P^i$ with $w^i_j=0$, $\alpha^{i+1}>0$.
\end{proof}
\vspace{12 pt}

\begin{lemma}\label{za}
If $w^i=x^i$ but $w^i$ is suboptimal for (\ref{4.2}), then $S(w^{i+2})>S(w^i)\text{ }\forall \alpha\in (0,1]$.\\
\end{lemma}

\begin{proof} of Lemma \ref{za}\\
{
\begin{align*}
\label{6}
S(w^{i+2})&=S(w^{i+1}+\alpha(x^{i+1}-w^{i+1}))=S(w^i+\alpha(x^{i+1}-w^{i}))\\
&=\frac{(w^i+\alpha(x^{i+1}-w^{i}))^Te}{\sqrt{(w^i+\alpha(x^{i+1}-w^{i}))^T\Sigma(w^i+\alpha(x^{i+1}-w^{i}))}}\\
&=\frac{(w^i+\alpha(x^{i+1}-w^{i}))^Te}{\norm{\Sigma^\frac{1}{2}(w^i+\alpha(x^{i+1}-w^{i}))}_2}
\end{align*}

Focusing on the numerator,

\begin{align*}
(w^i+\alpha(x^{i+1}-w^{i}))^Te&=(1-\alpha)w^{iT}e+\alpha x^{i+1T}e\\
&=(1-\alpha)\left(\Sigma^{-1}_{P^{i+1}}\left[\begin{array}{c}e_{P^i}\\e_k+\hat{\mu}^i_k\\\end{array}\right]\right)^T\left[\begin{array}{c}e_{P^i}\\e_k\\\end{array}\right]+\alpha x^{i+1T}e\\
&=(1-\alpha)\left[\begin{array}{c}e_{P^i}\\e_k+\hat{\mu}^i_k\\\end{array}\right]^Tx^{i+1}_{P^{i+1}}+\alpha x^{i+1T}e\\
&=x^{i+1T}e+(1-\alpha)\hat{\mu}^i_kx^{i+1}_k
\end{align*}

Focusing on the denominator,

\begin{align*}
\norm{\Sigma^\frac{1}{2}(w^i+\alpha(x^{i+1}-w^{i}))}_2
&=\norm{\Sigma^\frac{1}{2}((1-\alpha)w^i+\alpha x^{i+1})}_2\\
&=\norm{\Sigma^\frac{1}{2}_{P^{i+1}}\left((1-\alpha)\Sigma^{-1}_{P^{i+1}}\left[\begin{array}{c}e_{P^i}\\e_k+\hat{\mu}^i_k\\\end{array}\right]+\alpha x^{i+1}_{P^{i+1}}\right)}_2\\
&=\norm{\Sigma^\frac{1}{2}_{P^{i+1}}\left(x^{i+1}_{P^{i+1}}+(1-\alpha)\Sigma^{-1}_{P^{i+1}}
\left[\begin{array}{c}\mathbf{0}\\\hat{\mu}^i_k\\\end{array}\right]\right)}_2\\
&=\sqrt{x^{i+1T}\Sigma x^{i+1}+2(1-\alpha)x^{i+1}_k\hat{\mu}^i_k+(1-\alpha)^2(\hat{\mu}^i_k)^2\Sigma^{-1}_{kk}}\\
&=\sqrt{x^{i+1T}\Sigma x^{i+1}+2(1-\alpha)x^{i+1}_k\hat{\mu}^i_k-(1-\alpha)^2x^{i+1}_k\hat{\mu}^i_k}\\
&=\sqrt{x^{i+1T}\Sigma x^{i+1}+(1-\alpha^2)\hat{\mu}^i_kx^{i+1}_k}\\
\end{align*}

Therefore,

\begin{align*}
S(w^{i+2})&=\frac{x^{i+1T}e+(1-\alpha)\hat{\mu}^i_kx^{i+1}_k}{\sqrt{x^{i+1T}\Sigma x^{i+1}+(1-\alpha^2)\hat{\mu}^i_kx^{i+1}_k}}\\
&\geq\frac{x^{i+1T}e+(1-\alpha^2)\hat{\mu}^i_kx^{i+1}_k}{\sqrt{x^{i+1T}e+(1-\alpha^2)\hat{\mu}^i_kx^{i+1}_k}}\\
&=\sqrt{x^{i+1T}e+(1-\alpha^2)\hat{\mu}^i_kx^{i+1}_k}\\
&>\sqrt{x^{i+1T}e+\hat{\mu}^i_kx^{i+1}_k}\\
&=\sqrt{w^{iT}e}=\frac{w^{iT}e}{\sqrt{w^{iT}e}}=\frac{w^{iT}e}{\sqrt{w^{iT}_{P^i}\Sigma_{P^i}\Sigma^{-1}_{P^i}e_{P^i}}}\\
&=\frac{w^{iT}e}{\sqrt{w^{iT}_{P^i}\Sigma_{P^i}w^i_{P^i}}}=\frac{w^{iT}e}{\sqrt{w^{iT}\Sigma w^i}}=S(w^i)
\end{align*}

}%
\end{proof}

The algorithm is monotone increasing as for all $i$, $S(x^i)\geq S(w^i)$, and by the quasi-concavity of $S(w)$, this implies that $S(w^i+\alpha(x^i-w^i))\geq S(w^i)$. If $w^0$ is not optimal $S(w^2)>S(w^0)$ by Lemmas \ref{xk} and \ref{za}, and since $S(w^0)\geq S(w)$ for all portfolios $w$ of size 1, $|P^i|\geq 2$ for all $i\geq 1$.\\

Assume now $x^i\neq w^i$, and this holds until $x^{i+m}=w^{i+m}$, where $m\leq n-2$ and $|P^{i+m}|\leq n-m$. If the solution is not optimal, there exists $q\leq n-m-2$ indices of $P^{i+m}$ such that $w^{i+m+1}_j=0$ and $p^{i+m+1}_j<0$.
After $q$ iterations, there exists no $j\in P^{i+m+1+q}$ such that $w^{i+m+1+q}_j=0$ and $p^{i+m+1+q}_j<0$, so by Lemma \ref{xk}, $\alpha^{i+m+1+q}>0$ and by Lemma \ref{za}, $S(w^{i+m+2+q})>S(w^{i+m+q})\geq S(w^i)$. Assuming $m=0$ considers the case where $x^i=w^i$, and so we have shown that the algorithm is strictly increasing after $m+2+q\leq n$ iterations.\\

Since the optimal value is bounded and the algorithm is strictly monotone increasing over intervals of $n$ iterations, the algorithm converges.
\end{proof}

\begin{proof} of Proposition \ref{skew}\\
{
Beginning from equation (\ref{eq:omega}) of the proof of Theorem \ref{main},
\begin{equation}
\Omega(R )=\frac{\int_{L}^{\infty }\int_{L}^{r } f(r)dxdr}{\int_{-\infty}^{L}\int_{r}^{L }        f(r)dxdr}=\frac{\int_{L}^{\infty }(r-L)f(r)dr}{\int_{-\infty}^{L}(L-r)f(r)dr}=
\frac{\BE{R}-L}{\BE{(L-R)^+}}-1.
\nonumber
\end{equation}
}
\end{proof}

\end{document}